\pdfoutput=1
\documentclass{cccg11}
\usepackage{graphicx,amssymb,amsmath}

\usepackage[utf8]{inputenc}

\usepackage{pdfsync}
\usepackage{color}
\usepackage{subfigure}
\usepackage{url}
\usepackage{graphicx}

\usepackage{bm}
\usepackage{latexsym}
\usepackage{amsfonts}
\usepackage{hyperref}

\usepackage[square,numbers,sort&compress]{natbib}

\newcommand{\figref}[1]{Figure \ref{fig:#1}}
\newcommand{\lemref}[1]{Lemma \ref{lem:#1}}

\newcommand{\theoref}[1]{Theorem \ref{theo:#1}}
\newcommand{\secref}[1]{Section \ref{sec:#1}}
\newcommand{\appref}[1]{Appendix \ref{sec:#1}}
\newcommand{\corref}[1]{Corollary \ref{cor:#1}}

\newcommand{\lemlab}[1]{\label{lem:#1}}

\newcommand{\theolab}[1]{\label{theo:#1}}
\newcommand{\seclab}[1]{\label{sec:#1}}
\newcommand{\corlab}[1]{\label{cor:#1}}

\newcommand{\Z}{\mathbb{Z}}

\newcommand{\bgamma}{\bm{\gamma}}

\title{Rigid components in fixed-lattice and cone frameworks}

\author{Matthew Berardi\thanks{Department of Mathematics, Temple University,
{\tt \{mberardi,justmale,theran\}@temple.edu}}
\and
Brent Heeringa\thanks{Department of Computer Science,
Williams College, {\tt heeringa@cs.williams.edu}}
\and
Justin Malestein\footnotemark[1]
\and
Louis Theran\footnotemark[1]}

\index{Berardi, Matthew}
\index{Heeringa, Brent}
\index{Malestein, Justin}
\index{Theran, Louis}

\begin{document}
\thispagestyle{empty}
\maketitle

\begin{abstract}
We study the fundamental algorithmic rigidity problems for generic frameworks
periodic with respect to a fixed lattice or a finite-order rotation in the plane.
For fixed-lattice frameworks we give an $O(n^2)$ algorithm for deciding generic rigidity and
an $O(n^3)$ algorithm for computing rigid components.  If the order of rotation
is part of the input, we give an $O(n^4)$ algorithm for deciding rigidity; in the case
where the rotation's order is $3$, a more specialized algorithm solves all the
fundamental algorithmic rigidity problems in $O(n^2)$ time.
\end{abstract}

\section{Introduction}\seclab{intro}
The geometric setting for this paper involves two variations on the  well-studied
\emph{planar bar-joint} rigidity model: \emph{fixed-lattice periodic frameworks} and
\emph{cone frameworks}.  A \emph{fixed-lattice periodic framework} is an infinite structure,
periodic with respect to a lattice, where the allowed continuous motions preserve,
the lengths and connectivity of the bars, as well as the periodicity with respect to a fixed lattice.
See \figref{fixedlattice}(a) for an example.  A \emph{cone framework} is also made of
fixed-length bars connected by universal joints, but it is finite and symmetric with
respect to a finite order rotation; the allowed continuous motions preserve the bars' lengths
and connectivity and symmetry with respect to a fixed rotation center.
Cone frameworks get their name from the fact
that the quotient of the plane by a finite order rotation is a flat cone with opening angle
$2\pi/k$ and the quotient framework, embedded in the cone with geodesic ``bars'',
captures all the geometric information \cite{MT11}.  \figref{cone}(a) shows an example.

A fixed-lattice framework is \emph{rigid} if the only allowed motions are translations and \emph{flexible} otherwise.
A cone-framework is \emph{rigid} if the only allowed motions are rotations around the center and \emph{flexible} otherwise.
The alternate formulation for cone frameworks
says that rigidity means the only allowed motions are isometries of the cone, which is just rotation
around the cone point.  A framework is \emph{minimally rigid} if it is rigid, but ceases to be so if any
of the bars are removed.

\paragraph{Generic rigidity} The combinatorial model for the fixed-lattice
and cone frameworks introduced above is given by a \emph{colored graph} $(G,\bgamma)$: $G=(V,E)$ is
a finite directed graph and $\bgamma=(\gamma_{ij})_{ij\in E}$ is an assignment of a group element
$\gamma_{ij}\in \Gamma$ (the ``color'') to each edge $ij$ for a group $\Gamma$. For fixed-lattice frameworks,
the group $\Gamma$ is $\Z^2$, representing translations;
for cone frameworks it is $\Z/k\Z$ with $k\ge 2$ a natural number. See \figref{fixedlattice}(b) and \figref{cone}(b).

The colors can be seen as efficiently encoding a map $\rho$ from the oriented cycle space of $G$ into $\Gamma$;
$\rho$ is defined, in detail, in \secref{prelim}.  If the image of $\rho$ restricted to a subgraph $G'$ contains only the
identity element, we define the \emph{$\Gamma$-image} of $\rho$ to be \emph{trivial} otherwise it is \emph{non-trivial}.
\begin{figure}[htbp]
\centering
\subfigure[]{\includegraphics[width=.45\columnwidth]{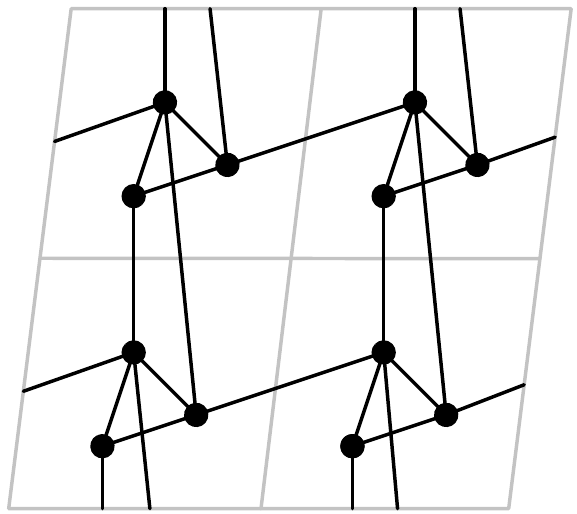}}
\subfigure[]{\includegraphics[width=.45\columnwidth]{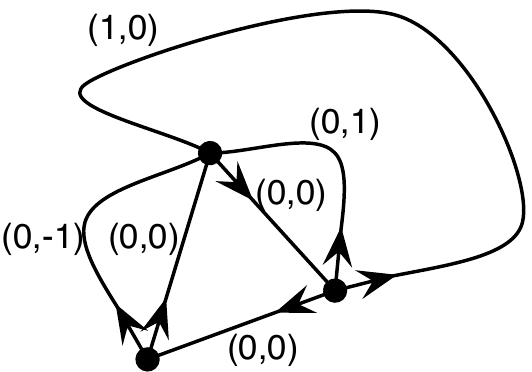}}
\caption{Periodic frameworks and colored graphs: (a) part of a periodic framework, with the representation of the
integer lattice $\Z^2$	shown in gray and the bars shown in black; 	(b) one possibility for the the associated colored
graph with $\Z^2$ colors on the edges.  (Graphics from \cite{MT10}.)}
\label{fig:fixedlattice}
\end{figure}
\begin{figure}[htbp]
\centering
\subfigure[]{\includegraphics[width=.65\columnwidth]{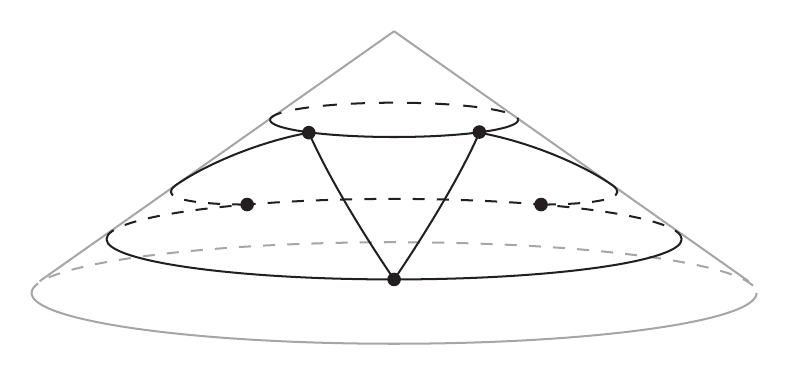}}
\subfigure[]{\includegraphics[width=.4\columnwidth]{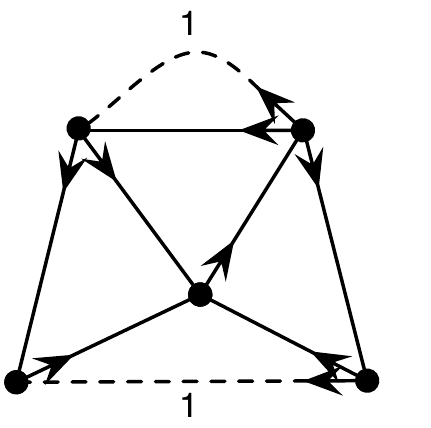}}
\subfigure[]{\includegraphics[width=.45\columnwidth]{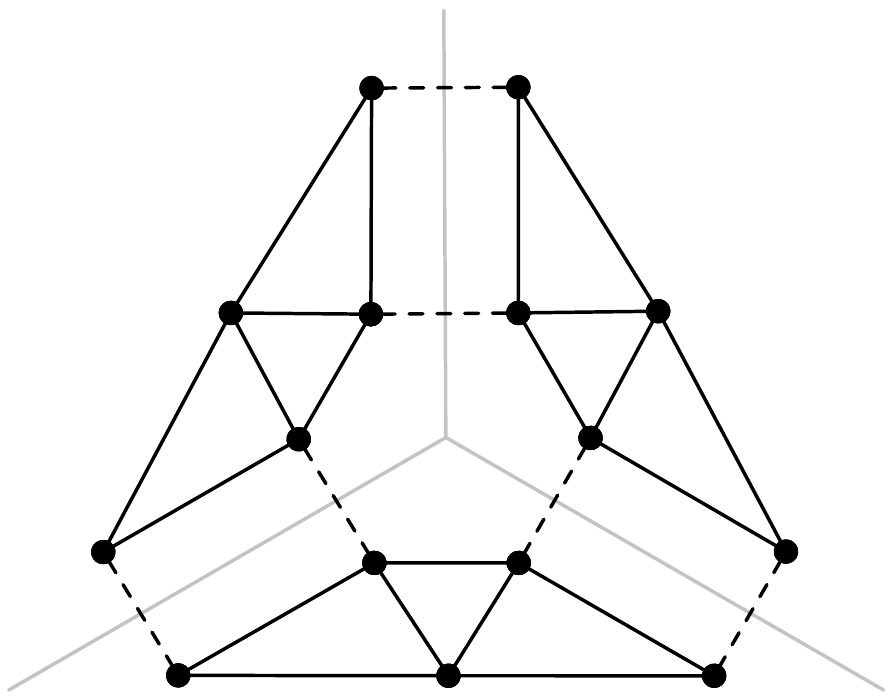}}
\caption{Cone-Laman graphs: (a) a realization of the framework on a cone
with opening angle $2\pi/3$ (graphic from Chris Thompson); (b) a $\Z/3\Z$-colored
graph (edges without colors have color $0$); (c)
the developed graph with $\Z/3\Z$-symmetry
(dashed edges are lifts of dashed edges in (b)).}
\label{fig:cone}
\end{figure}
The generic rigidity theory of planar frameworks with, more generally,
crystallographic symmetry has seen a lot of progress recently
\cite{BS10,R09,MT10,MT11}.  Elissa Ross \cite{R09} announced the following theorem:
\begin{theorem}[\cite{R09,MT10}]\theolab{fixedlattice}
A generic fixed-lattice periodic framework with associated colored graph $(G,\bgamma)$ is minimally rigid if and only if:
(1) $G$ has $n$ vertices and $2n-2$ edges; (2) all non-empty subgraphs $G'$ of $G$ with $m'$ edges and $n'$
vertices and trivial $\Z^2$-image satisfy $m'\le 2n'-3$; (3)  all non-empty subgraphs $G'$ with non-trivial
$\Z^2$-image satisfy $m'\le 2n'-2$.
\end{theorem}
The colored graphs appearing in the statement of \theoref{fixedlattice} are defined to be \emph{Ross graphs}; if
only conditions (2) and (3) are met, $(G,\bgamma)$ is \emph{Ross-sparse}. Ross graphs generalize the well-known
\emph{Laman graphs} which are uncolored, have $m=2n-3$ edges, and satisfy (2).
By \theoref{fixedlattice} the maximal rigid sub-frameworks of a generic fixed-lattice framework on a Ross-sparse colored graph $(G,\bgamma)$
correspond to maximal subgraphs of $G$ with $m'=2n'-2$; we define these to be the \emph{rigid components} of
$(G,\bgamma)$.  In the sequel, we will also refer to graphs with the Ross property for $\Gamma=\Z/k\Z$ as simply ``Ross graphs''.

Malestein and Theran \cite{MT11} proved a similar statement for cone frameworks:
\begin{theorem}[\cite{MT11}]\theolab{cone}
A generic cone framework with associated colored graph $(G,\bgamma)$ is minimally rigid if and only if:
(1) $G$ has $n$ vertices and $2n-1$ edges; (2) all non-empty subgraphs $G'$ of $G$ with $m'$ edges and $n'$
vertices and trivial $\Z/k\Z$-image satisfy $m'\le 2n'-3$; (3)  all non-empty subgraphs $G'$ with non-trivial
$\Z/k\Z$-image satisfy $m'\le 2n'-1$.
\end{theorem}
The graphs appearing in the statement of \theoref{cone} are called \emph{cone-Laman graphs}.  We define
\emph{cone-Laman-sparse} colored graphs and their rigid components similarly to the analogous definitions
for Ross-sparse graphs, with $2n'-1$ replacing $2n'-2$.

Ross and cone-Laman graphs are examples of the ``$\Gamma$-graded-sparse'' colored graphs
introduced in \cite{MT10,MT11}.  They are all matroidal families \cite{MT10,MT11},
which guarantees that greedy algorithms work correctly on them.

\paragraph{Main results}
In this paper we begin the investigation of the algorithmic theory of crystallographic rigidity
by addressing the fixed-lattice and cone frameworks.  Given a colored graph $(G,\bgamma)$,
we are interested in the rigidity properties of an associated generic framework.  Lee and Streinu
\cite{LS08} define three fundamental algorithmic rigidity questions:
\textbf{Decision} \emph{Is the input rigid?};
\textbf{Extraction} \emph{Find a maximum subgraph of the input corresponding to independent
length constraints};
\textbf{Components} \emph{Find the maximal rigid sub-frameworks of a flexible input}.

We give algorithms for these problems with running times shown in the following table

\vspace{2mm}
\begin{footnotesize}
\begin{tabular}{lccc}
\hline

& \textbf{Decision} & \textbf{Extraction} & \textbf{Components} \\
\hline
\hline
Fixed-lattice & $O(n^2)$ & $O(n^3)$ & $O(n^3)$ \\
\hline
Cone $k\neq 3$ & $O(n^4)$ & $O(n^5)$ & $O(n^5)$ \\
\hline
Cone $k=3$ & $O(n^2)$ & $O(n^2)$ & $O(n^2)$ \\
\hline
\end{tabular}
\end{footnotesize}
\vspace{-2mm}

\paragraph{Novelty}
Previously, the only known efficient combinatorial algorithms for any of these problems were pointed out in
\cite{MT10,MT11}: the Edmonds Matroid Union algorithm yields an algorithm with running times
$O(n^4)$ for \textbf{Decision} and $O(n^5)$ \textbf{Extraction}.  A folklore randomized  algorithm
based on Gaussian elimination gives an $O(n^3\operatorname{polylog}(n))$ algorithm for
\textbf{Decision} and \textbf{Extraction} of most rigidity problems, but this
doesn't easily generalize to \textbf{Components}.

The $O(n^2)$ running time for \textbf{Decision} for fixed-lattice frameworks
equals that from the \emph{pebble game} \cite{LS08,HJ97,BJ03} for the corresponding
problem in finite frameworks.  Although there are faster \textbf{Decision} algorithms \cite{GW92}
for finite frameworks, the pebble game is the standard tool in the field due to its elegance and
ease of implementation.  Our algorithms for cone frameworks with order $3$ rotation are a reduction
to the pebble games of \cite{LS08,HJ97,BJ03}.

The $O(n^3)$ running time for \textbf{Extraction} and \textbf{Components} in fixed-lattice frameworks is
worse by a factor of $O(n)$ than the pebble games for finite frameworks.  However, it is equal to the
$O(n^3)$ running time from \cite{LS08} for the ``redundant rigidity'' problem.  Computing
\emph{fundamental Laman circuits} (definition in \secref{prelim}) plays an important role
(though for different reasons) in both of these algorithms.

\paragraph{Roadmap and key ideas}
Our main contribution is a pebble game algorithm for Ross graphs, from which we can deduce
the corresponding results for general cone-Laman graphs.  Intuitively, the algorithmic
rigidity problems should be harder for Ross graphs than for Laman graphs, since the number of
edges allowed in a subgraph depends on whether the $\Z^2$-image of the subgraph is trivial or
not.  To derive an efficient algorithm we use three key ideas
(detailed definitions are given in \secref{prelim}):
\begin{itemize}
\item The Lee-Streinu-Theran \cite{LST07} approach of playing several copies of the pebble game for $(k,\ell)$-graphs
\cite{LS08}	with different parameters to handle different sparsity counts for different types of subgraphs.
\item A \emph{new} structural characterization of the edge-wise minimal colored graphs which violate the Ross counts
(\secref{comb}).
\item A \emph{linear time} algorithm for computing the $\Gamma$ image of a given subgraph (\secref{z2rank}).
\end{itemize}

Our algorithms for general cone-Laman graphs then use the Ross graph \textbf{Decision} algorithm as a subroutine.
When the order of the rotation is $3$, we can reduce the cone-Laman rigidity questions to Laman graph rigidity
questions directly, resulting in better running times.

\paragraph{Motivation}
Periodic frameworks, in which the lattice \emph{can} flex, arise in the study of \emph{zeolites},
a class of microporous crystals with a wide variety of industrial applications, notably in petroleum refining.
Because zeolites exhibit \emph{flexibility} \cite{SWTT06}, computing the degrees of freedom in \emph{potential}
\cite{R06,TRBR04} zeolite structures is a well-motivated algorithmic problem.

\paragraph{Other related work}
The general subject of periodic and crystallographic
rigidity has seen a lot of progress recently, see \cite{JOP10} for a list of announcements.
Bernd Schulze \cite{S10} has studied Laman graphs with a free $\Z/3\Z$ action in a different context.

\section{Preliminaries}\seclab{prelim}
In this section, we introduce the required background in colored graphs, hereditary sparsity, and
introduce a data structure for least common ancestor queries in trees that is an essential tool for us.

\paragraph{Colored graphs and the map $\rho$}  A pair $(G,\bgamma)$ is defined to be a \emph{colored graph} with
$\Gamma$ a group, $G=(V,E)$  a finite, directed graph on $n$ vertices and $m$ edges, and $\bgamma=(\gamma_{ij})_{ij\in E}$
is an assignment of a group element $\gamma\in \Gamma$ to each edge.

Let  $(G,\bgamma)$ be a colored graph, and let $C$ be a cycle in $G$ with a fixed traversal order.  We define $\rho(C)$ to be
\[
\rho(C) = 
\sum_{\substack{ij\in C \\ \text{$ij$ traversed} \\ \text{ forwards}}} \gamma_{ij}
-
\sum_{\substack{ij\in C \\ \text{$ij$ traversed} \\ \text{backwards}}} \gamma_{ij} 
\]
Since $\Gamma$ is always abelian in this paper,
we need not be concerned with the particular order of summation; our notation doesn't capture the specific traversal of $C$,
but this is not important here since we are interested in whether  $\rho(C)$ is trivial or not, which doesn't depend on sign.  For a subgraph $G'$
of $G$, we define $\rho(G')$ to be \emph{trivial} if its image on cycles spanned by $G'$ contains only the identity and \emph{non-trivial} otherwise.
We need the following fact about $\rho$.
\begin{lemma}[{\cite[Lemma 2.2]{MT10}}]\lemlab{fundamentalcycles}
Let $(G,\bgamma)$ be a colored graph.  Then $\rho(G)$ is trivial if and only if, for \emph{any} spanning forest $T$ of $G$, $\rho$ is trivial
on every fundamental cycle induced by $T$.
\end{lemma}

\paragraph{$(k,\ell)$-sparsity and pebble games}
The hereditary sparsity counts defining Ross and cone-Laman graphs generalize to \emph{$(k,\ell)$-sparse graphs} which satisfy
``$m'\le kn'-\ell$'' for all subgraphs; if in addition the total number of edges is $m=kn-\ell$, the graph is a \emph{$(k,\ell)$-graph}.
We also need the notion of a \emph{$(k,\ell)$-circuit}, which is an edge-minimal graph that is not $(k,\ell)$-sparse; these are always
$(k,\ell-1)$-graphs \cite{LS08}.  If $G$ is any graph, a \emph{$(k,\ell)$-basis} of $G$ is a maximal subgraph that is $(k,\ell)$-sparse;
if $G'$ is a $(k,\ell)$-basis of $G$ and $ij\in E(G) - E(G')$, the \emph{fundamental $(k,\ell)$-circuit of $ij$ with respect to $G'$}
is the unique (see \cite{LS08}) $(k,\ell)$-circuit in $G'+{ij}$.  See \cite{LS08} for a detailed development of this theory.
As is standard in the field, we use ``$(2,3)$-'' and ``Laman'' interchangeably.

Although $(k,\ell)$-sparsity is defined by exponentially many inequalities on subgraphs,
it can be checked in quadratic time using the \emph{pebble game} \cite{LS08}, an elegant incremental
approach that builds a $(k,\ell)$-sparse graph $G$ one edge at a time. Here, we will use the
pebble game as a ``black box'' to:
(1) Check if an edge $ij$ is in the span of any $(k,\ell)$-component of $G$ in $O(1)$ time \cite{LS08,LST05};
(2) Assuming that $G$ plus a new edge $ij$ is $(k,\ell)$-sparse, add the edge $ij$ to $G$ and update the components in amortized $O(n^2)$ time \cite{LS08};
(3) Compute the fundamental circuit with respect to a given $(k,\ell)$-sparse graph $G$ in $O(n)$ time \cite{LS08}.

\paragraph{Least common ancestors in trees} Let $T$ be a rooted tree with root $r$ and $i$ and $j$ be any vertices in $T$.
The \emph{least common ancestor} (shortly, LCA) of $i$ and $j$ is defined to be the vertex where the (unique, since $T$ is a tree)
paths from $i$ to $r$ and $j$ to $r$ first converge.  If either $i$ or $j$ is $r$, then this is just $r$.  A fundamental result of
Harel and Tarjan \cite{HT84} is that LCA queries can be answered in $O(1)$ time after
$O(n)$ preprocessing.

\section{Combinatorial lemmas}\seclab{comb}
In this section we prove structural properties of Ross and cone-Laman graphs that
are required by our algorithms.

\paragraph{Ross graphs}
Let $(G,\bgamma)$ be a colored graph and suppose that $G$ is a $(2,2)$-graph.
We can verify that $(G,\bgamma)$ is Ross by checking the $\Z^2$-images of a
relatively small set of subgraphs.
\begin{lemma}\lemlab{circuits}
Let $(G,\bgamma)$ be a colored graph and suppose that $G$ is a $(2,2)$-graph.  Then $(G,\bgamma)$ is a Ross
graph if and only if for \emph{any} Laman basis $L$ of $G$, the fundamental Laman circuit with respect to $L$
of every edge $ij\in E - E(L)$ has non-trivial $\Z^2$-image.
\end{lemma}
\figref{circuits} shows two examples.  The important point is that we can pick \emph{any} Laman basis $L$ of $G$.  The
proof is deferred to \appref{circuits}.  The main idea is that $G$ being a $(2,2)$-graph forces
all Laman circuits to be edge-disjoint, from which we can deduce all of them are fundamental
Laman circuits of every Laman basis.
\begin{figure}[htbp]
\centering
\subfigure[]{\includegraphics[width=.45\columnwidth]{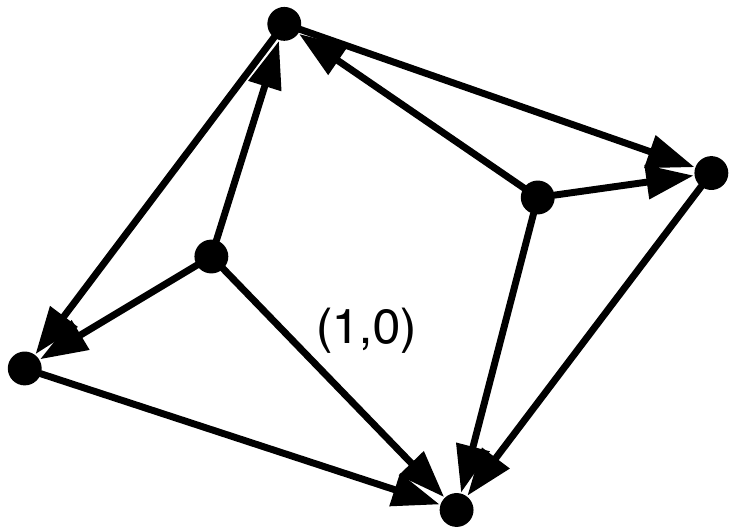}}
\subfigure[]{\includegraphics[width=.45\columnwidth]{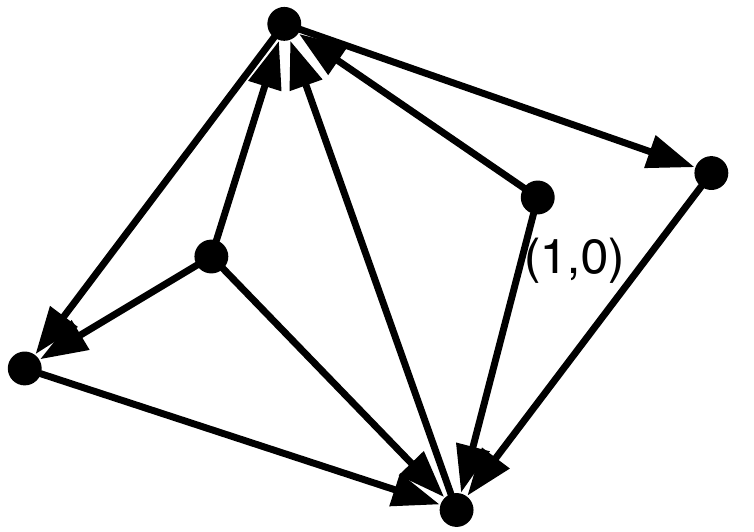}}
\caption{Examples of Ross and non-Ross graphs (edges without colors have color $(0,0)$):
(a) a Ross graph; the underlying graph is itself a Laman circuit; (b)
the underlying graph is a $(2,2)$-graph, but the uncolored $K_4$ subgraph has trivial image, so this
is not a Ross graph.  Note that $K_4$ is a Laman circuit, illustrating \lemref{circuits}}
\label{fig:circuits}
\end{figure}

\paragraph{Cone-Laman graphs}
Because cone-Laman graphs have an underlying $(2,1)$-graph, the statement of \lemref{circuits},
with $(2,1)$- replacing $(2,2)$- does \emph{not} hold for cone-Laman graphs.
\figref{21circuit} shows a counterexample.  The analogous statement, proven in \appref{conecircuits} is:
\begin{lemma}\lemlab{conecircuits}
Let $(G,\bgamma)$ be a colored graph.  Then $(G,\bgamma)$ is a cone-Laman graph if
and only if: (1) $G$ is a $(2,1)$-graph; (2) for \emph{any} $(2,2)$-basis $R$ of $G$, the
fundamental $(2,2)$-circuit $G'$ with respect to $R$ of $ij\in E(G)-E(R)$ becomes a Ross
graph after removing \emph{any} edge from $G'$; (3) for \emph{any} Laman-basis $L$ of $G$,
the fundamental Laman-circuits with respect to $L$ have non-trivial $\Gamma$-image.
\end{lemma}

\paragraph{Order three rotations}
In the special case where the group $\Gamma=\Z/3\Z$,
which corresponds to a cone with opening angle $2\pi/3$, we can give a simpler characterization of
cone-Laman graphs in terms of their \emph{development}.  The development
$\tilde{G}$  is defined by the following construction:
$\tilde{G}$ has three copies of each vertex $i$: $i_0$, $i_1$ and $i_2$;
a directed edge $ij$ with color $\gamma$ then generates three undirected edges
$i_{k}j_{k+\gamma}$ (addition is modulo 3).  See \figref{cone}(c)) for an
example. The development has a free $\Z/3\Z$-action; a subgraph of $\tilde{G}$ is defined to be \emph{symmetric}
if it is fixed by this action. In \appref{orderthree} we prove.
\begin{lemma}\lemlab{orderthree}
Let $(G,\bgamma)$ be a colored graph with $\Gamma=\Z/3\Z$.  Then $(G,\bgamma)$ is a cone-Laman graph
if and only if its development $\tilde{G}$ is a Laman graph.  Moreover, the rigid components of $(G,\bgamma)$
correspond to the symmetric rigid components of $\tilde{G}$.
\end{lemma}

\section{Computing the $\Gamma$-image of $\rho$}\seclab{z2rank}
We now focus on the problem of deciding whether the $\Gamma$-image of the map $\rho$,
defined in \secref{prelim}, is trivial on a colored
graph $(G,\bgamma)$.  The case in which $G$ is not connected follows easily by considering
connected components one at a time, so we assume from now on that $G$ is connected.
Let $(G,\bgamma)$ be a colored graph and $T$ be a spanning tree of $G$ with root $r$.  For a vertex $i$, there is a
unique path $P_i$ in $T$ from $r$ to $i$.  We define $\sigma_{ri}$ to be
\[
\sigma_{ri} = \sum_{\substack{jk\in P_i \\ \text{$jk$ traversed forwards}}} \gamma_{jk} -
\sum_{\substack{jk\in P_i \\ \text{$jk$ traversed backwards}}} \gamma_{jk}
\]
The notation $\sigma_{ri}$ extends in a natural way: for a a vertex $j$ on $P_i$, we define $\sigma_{ij}$ to be $\sigma_{ri} - \sigma_{rj}$;
if $\sigma_{ji}$ is defined, we define $\sigma_{ij} = -\sigma_{ji}$.  The key observation is the following lemma.

\begin{lemma}\lemlab{rholca}
Let $(G,\bgamma)$ be a connected colored graph, let $T$ be a rooted spanning tree of $G$, let $ij$ be an edge of $G$ not in $T$, and let $a$
be the least common ancestor of $i$ and $j$.  Then, if $C$ is the fundamental cycle of $ij$ with respect to $T$,
$\rho(C) = \sigma_{ai} + \gamma_{ij} - \sigma_{ja}$.
\end{lemma}
\begin{proof}
Traversing the fundamental cycle of $ij$ so that $ij$ is crossed from $i$ to $j$ means: going from $i$ to $j$,
from $j$ to the LCA $a$ of $i$ and $j$ towards the root, and then from $a$ to $i$ away from the root.
\end{proof}

We now show how to compute whether the $\Gamma$-image of a colored graph is trivial in linear time.  The idea
used here is closely related to a folklore $O(n^2)$ algorithm for all-pairs-shortest paths in trees.\footnote{We thank
David Eppstein for clarifying the tree APSP trick's origins on MathOverflow.}
\begin{lemma}\lemlab{image}
Let $(G,\bgamma)$ be a connected colored graph with $n$ vertices and $m$ edges.  There is an $O(n+m)$ time
algorithm to decide whether the $\Gamma$-image of $\rho(G)$ is trivial.
\end{lemma}
The rest of this section gives the proof of \lemref{image}.  We first present the algorithm.
\\ \textbf{Input:} A colored graph $(G,\bgamma)$ \\
\textbf{Question:} Is $\rho(G)$ trivial? \\
\textbf{Method:}
\begin{itemize}
\item Pick a spanning tree $T$ of $G$ and root it.
\item Compute $\sigma_{ri}$ for each vertex $i$ of $G$.
\item For each edge $ij$ not in $T$, compute the image of its fundamental cycle in $T$.
\item Say `yes' if any of these images are not the identity and `no' otherwise.
\end{itemize}

\paragraph{Correctness} This is an immediate consequence of \lemref{fundamentalcycles}, since the algorithm checks
all the fundamental cycles with respect to a spanning tree.

\paragraph{Running time}  Finding the spanning tree with BFS is $O(m)$ time, and once the tree
is computed, the $\sigma_{ri}$ can be computed with a single pass over it in $O(n)$ time.
\lemref{rholca} says that the image of any fundamental cycle with respect to $T$ can
be computed in $O(1)$ time once the LCA of the endpoints of the non-tree edge is known.
Using the Harel-Tarjan data structure, the total cost of LCA queries is $O(n+m)$, and the
running time follows.

\paragraph{The pebble game for Ross graphs}\seclab{algo}
We have all the pieces in place to describe our algorithm for the rigidity
problems in Ross graphs.

\paragraph{Algorithm: Rigid components in Ross graphs} ~~ \\
\textbf{Input:} A colored graph $(G,\bgamma)$ with $n$ vertices and $m$ edges. \\
\textbf{Output:} The rigid components of $(G,\bgamma)$. \\
\textbf{Method:}  We will play the pebble game for $(2,3)$-sparse graphs and
the pebble game for $(2,2)$-sparse graphs in parallel.  To start,
we initialize each of these separately, including data structures for
maintaining the $(2,2)$- and $(2,3)$-components.

Then, for each colored edge $ij\in E$:
\begin{itemize}
\item[\textbf{(A)}] If $ij$ is in the span of a $(2,2)$-component in the $(2,2)$-sparse graph we are maintaining, we discard $ij$ and
proceed to the next edge.
\item[\textbf{(B)}] If $ij$ is not in the span of any $(2,3)$-component, we add $ij$ to both the $(2,2)$-sparse and $(2,3)$-sparse
graphs we are building, and update the components of each.
\item[\textbf{(C)}] Otherwise, we use the $(2,3)$-pebble game to identify the smallest $(2,3)$-block $G'$ spanning $ij$.  We add
$ij$ to this subgraph $G'$ and compute its $\Z^2$-image.  If this is trivial, we discard $ij$ and proceed to the next edge.
\item[\textbf{(D)}] If the image of $G'$ was non-trivial, add $ij$ to the $(2,2)$-sparse graph we are maintaining and update its
rigid components.
\end{itemize}

The output is the (2,2)-components in the $(2,2)$-sparse graph we built.

\paragraph{Correctness}  By definition, the rigid components of a Ross graph are its $(2,2)$-components.  Step \textbf{(A)}
ensures that we maintain a $(2,2)$-sparse graph; steps \textbf{(B)} and \textbf{(C)}, by \lemref{circuits} imply that when
new $(2,2)$-blocks are formed \emph{all} of them have non-trivial $\Z^2$-image, which is what is required for Ross-sparsity.
Step \textbf{(D)} ensures that the rigid components are updated at every step.  The matroidal property implies that a
greedy algorithm is correct.

\paragraph{Running time} By \cite{LS08,LST05}, steps \textbf{(A)}, \textbf{(B)}, and \textbf{(D)} require
$O(n^2)$ time over the entire run of the algorithm
(the analysis of the time taken to update components is amortized).  Step \textbf{(C)}, by \cite{LS08} and
\lemref{rholca} requires $O(n)$ time.  Since $\Omega(m)$ iterations may enter step \textbf{(C)}, this becomes the
bottleneck, resulting in an $O(nm)$ running time, which is $O(n^3)$.

\paragraph{Modifications for other rigidity problems}  We have presented and analyzed an algorithm for computing the rigid
components in Ross graphs.  Minor modifications give solutions to the \textbf{Decision} and \textbf{Extraction} problems.
For \textbf{Extraction}, we just return the $(2,2)$-sparse graph we built; the running time remains $O(n^3)$.  For
\textbf{Decision}, we simply stop and say `no' if any edge is ever discarded.  Since we process
at most $O(n)$ edges, the running time becomes $O(n^2)$.

\section{Pebble games for cone-Laman graphs} We now describe our algorithms for cone-Laman graphs.

\paragraph{Order-three rotations}
We start with the special case when the group $\Gamma=\Z/3\Z$.  In this case, the
following algorithm's correctness is immediate from \lemref{orderthree}.
The running time follows from \cite{LS08,LST05,BJ03} and the fact that the development can
be computed in linear time.

\textbf{Input:} A colored graph $(G,\bgamma)$ with $n$ vertices and $m$ edges. \\
\textbf{Output:} The rigid components of $(G,\bgamma)$. \\
\textbf{Method:}
\begin{itemize}
\item[\textbf{(A)}] Compute the development $\tilde{G}$ of $(G,\bgamma)$.
\item[\textbf{(B)}] Use the $(2,3)$-pebble game to compute the rigid components of $\tilde{G}$.
\item[\textbf{(C)}] Return the subgraphs of $G$ corresponding to the symmetric rigid components in $\tilde{G}$.
\end{itemize}

\paragraph{General cone-Laman graphs}
For colored graphs with $\Gamma=\Z/k\Z$, we don't have an analogue of \lemref{orderthree}, and the
development may not be polynomial size.  However, we can modify our pebble game
for Ross graphs to compute the rigid components.  Here is the algorithm:
\textbf{Input:} A colored graph $(G,\bgamma)$ with $n$ vertices and $m$ edges, and an integer $k$. \\
\textbf{Output:} The rigid components of $(G,\bgamma)$. \\
\textbf{Method:}  We initialize a $(2,1)$-pebble game, a $(2,2)$-pebble game,
and a $(2,3)$-pebble game.  Then, for each edge	$ij\in E(G)$:
\begin{itemize}
\item[\textbf{(A)}] If $ij$ is in the span of a $(2,1)$-component in the $(2,1)$-sparse graph we are maintaining,
we discard $ij$ and proceed to the next edge.
\item[\textbf{(B)}] If $ij$ is not in the span of any $(2,3)$-component, we add $ij$ to all three sparse
graphs we are building,  update the components of each, and proceed to the next edge.
\item[\textbf{(C)}] If $ij$ is not in the span of any $(2,2)$-component, we check that its fundamental
Laman circuit in the $(2,3)$-sparse graph has non-trivial $Z/k\Z$-image.  If not, discard $ij$.  Otherwise, add
$ij$ to the $(2,1)$- and $(2,2)$-sparse graphs and update components.
\item[\textbf{(D)}]  Otherwise $ij$ is not in the span of any $(2,1)$-component.  We find the minimal $(2,2)$-block $G'$
spanning $ij$ and check if $G'+ij$ becomes a Ross graph after removing \emph{any} edge.  If so, add $ij$ to the
$(2,1)$-graph we are building.  Otherwise discard $ij$.
\end{itemize}

The output is the $(2,1)$-components in the $(2,1)$-sparse graph we built.

\paragraph{Analysis}  The proof of correctness follows from \lemref{conecircuits} and an argument similar to the
one used to show that the pebble game for Ross graphs is correct.  Each loop iteration takes $O(n^3)$ time,
from which the claimed running times follow.

\section{Conclusions and remarks}\seclab{conclusions}
We studied the three main algorithmic rigidity questions for generic
fixed-lattice periodic frameworks and cone frameworks.  We gave algorithms based on the
pebble game for each of them.  Along the way we introduced several new ideas: a
linear time algorithm for computing the $\Gamma$-image of a colored graph, a
characterization of Ross graphs in terms of Laman circuits, and a characterization of cone-Laman graphs
in terms of the development for $k=3$ and Ross graphs for general $k$.

\paragraph{Implementation issues} The pebble game has become \emph{the} standard algorithm in the rigidity
modeling community because of its elegance, ease of implementation, and reasonable implicit constants.  The original data
structure of Harel and Tarjan \cite{HT84}, unfortunately, is too complicated to be of much use except as a theoretical tool.
More recent work of Bender and Farach-Colton \cite{BF00} gives a vastly simpler data structure for $O(1)$-time LCA that
is not much more complicated than the \emph{union pair-find} data structure of \cite{LST05} used in the pebble game.
This means that the algorithm presented here is implementable as well.

\small
% \bibliographystyle{abbrvnat}
% \bibliography{cccg}

\newpage

\appendix

\section{Details for \lemref{circuits}}\seclab{circuits}
In this appendix, we prove \lemref{circuits}.  We start off with some additional facts about
Laman graphs and circuits that are needed.

\paragraph{Additional facts about Laman graphs}
The matroidal property \cite[Theorem 2]{LS08} of Laman graphs implies that if $G$ is a graph with Laman basis $L$, \emph{any}
Laman circuit in $G$ can be generated by a sequence of ``circuit elimination'' steps starting from the fundamental Laman circuits
with respect to $L$; circuit elimination generates a new Laman circuit from two that overlap by discarding some edges from the
intersection.

The matroidal property implies that when all the Laman circuits in a graph are disjoint,
they are all fundamental circuits, independent of any choice of Laman basis.
\begin{lemma}\lemlab{elim}
Let $G$ be a graph and suppose that the Laman circuits in $G$ are all edge disjoint.  Then,
all Laman bases of $G$ have the same fundamental circuits, and every Laman circuit in $G$
is a fundamental circuit.
\end{lemma}
\begin{proof}
Let $L$ be a Laman basis of $G$.  The matroidal property implies that all the Laman circuits in $G$ are either
fundamental Laman circuits with respect to $L$ or can be generated by circuit elimination moves.  By
hypothesis, all the Laman circuits in $G$, and therefore all the fundamental circuits with respect to $L$,
are edge disjoint.  This means that there are no circuit elimination steps possible, forcing every Laman circuit in $G$
to be a fundamental circuit with respect to $L$.  This proves the second part of the Lemma.  Since $L$ was arbitrary,
the first part follows at once.
\end{proof}

\paragraph{Proof of \lemref{circuits}}
Let $(G,\bgamma)$ satisfy the assumptions of the lemma.  We start with the observation that every $(2,2)$-block $G'$ in $G$
must contain a Laman circuit:  a Laman basis for $G'$ cannot contain every edge of $G'$ (it has one too many), so there is a
fundamental Laman circuit with respect to this basis.  But then if \emph{any} $(2,2)$-block $G'$ in $G$ has trivial $\Z^2$-image, then
so do all its subgraphs, which must include a Laman circuit.  This implies that $(G,\bgamma)$ is Ross if and only if \emph{every} Laman
circuit has non-trivial $\Z^2$-image.

To complete the proof, we need to show that it is sufficient to restrict ourselves to the fundamental Laman circuits of
\emph{any} Laman basis $L$ of $G$.  To do this, we note that Laman circuits are $(2,2)$-blocks in $G$ that, by definition,
do not contain any strictly smaller $(2,2)$-blocks.  Now we make use of the hypothesis that $G$ is a $(2,2)$-graph:
the structure theorem for $(k,\ell)$-graphs \cite[Theorem 5]{LS08} says that \emph{any} pair of $(2,2)$-blocks in $G$ either has
no edge intersection or intersects on a $(2,2)$-block.  It then follows, since they can't contain smaller
$(2,2)$-blocks that the Laman circuits in $G$ are all edge disjoint.  \lemref{elim} now applies, completing
the proof.
\hfill $\square$

\section{Details for \lemref{conecircuits}}\seclab{conecircuits}
This appendix provides the proof of \lemref{conecircuits}.  Analogously to the case
of Ross graphs, $(G,\bgamma)$ will turn out to be cone-Laman if and only if $G$ is a
$(2,1)$-graph and all Laman circuits have non-trivial $\Gamma$-image.  The
difficulty, as illustrated in \figref{21circuit}, is that because $G$ is not
$(2,2)$-sparse, we \emph{can't} pick a Laman basis arbitrarily and then
look only at fundamental Laman circuits.

\begin{figure}[htbp]
\centering
\includegraphics[width=.75\columnwidth]{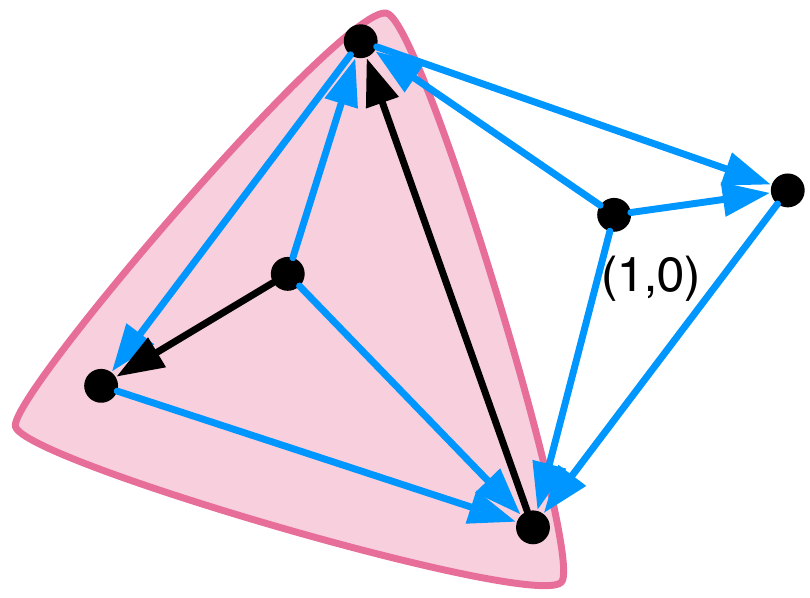}
\caption{A colored graph that is not cone-Laman: the underlying graph is a $(2,1)$-graph,
but there is a $K_4$ subgraph (indicated in pink)	with trivial $\Z^2$-image.
With respect to the Laman basis indicated by blue edges, it is \emph{not}
a fundamental circuit.}
\label{fig:21circuit}
\end{figure}
To get around this problem, we will reduce to the case when $G$ is a $(2,2)$-circuit;
i.e., a $(2,1)$-graph such that after removing \emph{any} edge from $G$, the result is a $(2,2)$-graph.
\begin{lemma}\lemlab{22circuit}
Let $(G,\bgamma)$ be a colored graph with $\Gamma=\Z/k\Z$, and $G$ a $(2,2)$-circuit.  Then $(G,\bgamma)$ is
cone-Laman if and only if removing \emph{any} edge from $G$ results in a Ross-graph.
\end{lemma}
\begin{proof}
One direction is straightforward: if there is some edge $ij$ such that removing from $G$ it leaves a
graph that is not Ross, then $G-{ij}$ must have some subgraph with trivial $\Gamma$-image that is not
Laman-sparse.  Since this subgraph is also a subgraph of $G$, this shows that $G$ is not cone-Laman.

For the other direction, we start by noting again that $(G,\bgamma)$ is cone-Laman if and only if every Laman
circuit in $G$ has non-trivial $\Gamma$-image.  Laman circuits are a subset of the $(2,2)$-blocks in $G$,
and we will show that every $(2,2)$-block in $G'$ has non-trivial $\Gamma$-image when the hypothesis of the Lemma
are met.  Let $G'$ be a $(2,2)$-block in $G$.  Since $G$ has one edge too many to be a $(2,2)$-graph, $G'$ is not
all of $G$.  Removing an edge $ij$ not in $G'$ leaves a subgraph $G-ij$ that is, by hypothesis, Ross, so $G'$
has non-trivial $\Gamma$-image.  Since $G'$ was arbitrary, we are done.
\end{proof}

In addition to the key \lemref{22circuit}, we also need two other additional facts about $(2,1)$-graphs.
\begin{lemma}\lemlab{22circuitsdisjoint}
Let $G$ be a $(2,1)$-graph.  Then the $(2,2)$-circuits in $G$ are edge disjoint.
\end{lemma}
\begin{proof}
This is a simple application of minimality of circuits, and the $(k,\ell)$-graph structure theorem
\cite[Theorem 5]{LS08}, similar to the case of Laman-circuits in a $(2,2)$-graph.
\end{proof}

\begin{lemma}\lemlab{22orlamancircuits}
Let $G$ be a $(2,1)$-graph, and let $G'$ be a Laman circuit in $G$.  Then either
$G'$ is contained in a $(2,2)$-circuit, or $G'$ is a fundamental Laman circuit
with respect to any Laman-basis of $G$.
\end{lemma}
\begin{proof}
Let $L$ be an arbitrary Laman basis and extend it to a $(2,2)$-basis $R$.  If $G'$ is edge
disjoint from all $(2, 2)$-circuits, then $G' \subset R$.  By the proof of \lemref{circuits},
any such Laman circuit is a fundamental circuit of $L$.

Suppose instead that $G'$ instersects a $(2,2)$-circuit $G''$ in at least one edge.
Let $n'$ and $n''$ be the number of vertices spanned by each subgraph, and
let $m'$ and $m''$ be the number of edges.  We define $n_\cup$, $n_\cap$, $m_\cup$, and $m_\cap$
similarly for the intersection and union of $G'$ and $G''$.
Because $G'$ is a $(2,2)$-graph, we get the sequence of inequalities
\begin{eqnarray}
2n_\cap - 2 & \ge m_\cap = 2n'-2+2n''-1 - m_\cup \\
& \ge 2n'-2+2n''-1 - 2n_\cup + 1 \\
& = 2n_\cap-2
\end{eqnarray}
Since any proper subgraph of $G'$ is Laman-sparse, we must have $G' \cap G'' = G'$.
\end{proof}

\paragraph{Proof of \lemref{conecircuits}}
It is enough to prove that every Laman circuit in $G$ has non-trivial $\Gamma$-image.
\lemref{22orlamancircuits} says that there are two types: those that don't intersect
any other Laman circuits, which are fundamental Laman circuits for any Laman basis; the
other type are all subgraphs of $(2,2)$-circuits, all of which are edge-disjoint by
\lemref{22circuitsdisjoint}.

The Lemma then follows by \lemref{22circuit}.
\hfill $\square$

\section{Details for \lemref{orderthree}}\seclab{orderthree}

In this appendix we prove \lemref{orderthree}. First we start with some
preliminaries, including a formal definition of the development and some facts
about $\Z/3\Z$-rank.

\paragraph{The development and covering map}
Let $(G,\bgamma)$ be a colored graph with colors in $\Z/3\Z$.  We define the \emph{development}
$\tilde{G}$ of $(G,\bgamma)$ to be the undirected graph resulting from
the following construction. For every vertex $i \in V(G)$,
$V(\tilde{G})$ has three elements, $i_0, i_1, i_2$ and for every edge $ij
\in E(G)$ (where $j$ is the head), $E(\tilde{G})$ has three elements,
$i_0j_{0+\gamma_{ij}}, i_1j_{1+\gamma_{ij}}, i_2j_{2+\gamma_{ij}}$ where
$\gamma_{ij}$ is the color of edge $ij$.
Arithmetic is performed modulo $3$. We observe that
$\tilde{G}$ has exactly three times as many edges and vertices as $G$.

Given a $\Z/3\Z$-colored graph $(G,\bgamma)$ and its development $\tilde{G}$, there
is a natural covering map $\pi: \tilde{G}\to G$ that sends $i_{\gamma}\in V(\tilde{G})$  to
$i\in V(G)$ and an edge $i_{\gamma_i}j_{\gamma_j}\in E(\tilde G)$ to $ij\in E(G)$.  The
pre-image $\pi^{-1}(i)$ is defined to be the \emph{fiber over $i$}; the fiber over an edge $ij$
is defined similarly.

\paragraph{Graphs with a free $\Z/3\Z$-action}
A \emph{graph automorphism} $\alpha$ of a graph $G=(V,E)$ is a bijection
between $V$ and itself that preserves edges; i.e., $\alpha: V\to V$ is an
automorphism if and only if $\alpha$ is a permutation and $ij\in E$
implies that  $\alpha(i)\alpha(j)$ is also in $E$.  The automorphisms
of $G$ naturally form a group.

A graph $G$ is defined to admit a \emph{free $\Z/3\Z$-action} if there is a
faithful representation of $\Z/3\Z$ by automorphisms $\alpha_i$, $i\in \{0,1,2\}$
of $G$ that act without fixed points, except for the identity.
If $G$ has a free $\Z/3\Z$-action and $G'$
is a subgraph of $G$ then the \emph{orbit} $\mathcal{O}(G')$ is defined to be
$\mathcal{O}(G') = G' \cup \alpha_1(G') \cup \alpha_2(G')$.  A subgraph $G'$
of $G$ is defined to be \emph{symmetric} if it coincides with its orbit.

\begin{lemma}
Let $(G,\bgamma)$ be a $\Z/3\Z$-colored graph and let $\tilde{G}$ be the development.
Then $\tilde{G}$ has a free $\Z/3\Z$-action.
\end{lemma}

\begin{proof}
Define $\alpha_z : V(\tilde{G})\to V(\tilde{G})$ to be $i_{\gamma}\mapsto i_{\gamma+z}$ for $z\in 0,1,2$.
These functions are clearly permutations of $V(\tilde{G})$ that have no fixed points, except for $\alpha_0$.
Since $\alpha_0(i) = \alpha_1(i)+\alpha_2(i)$ they represent $\Z/3\Z$.  To see that they are automorphisms,
we note that the the fibers of any vertex or edge of $G$ are closed under the action of the $\alpha_i$
by the definition of the developing map: this is clear for vertices and for
edges, if $i_{\gamma}j_{\gamma+\gamma_{ij}}$ is an edge of $\tilde{G}$, then $\tilde{G}$ also has an edge $i_{\gamma+z}j_{\gamma+\gamma_{ij}+z}$, which is also
in the fiber over $ij$.
\end{proof}

\paragraph{Facts about the development}
The essence of \lemref{orderthree} is that we can
read out sparsity properties and the $\Gamma$-image of
subgraphs of the colored graph $(G,\bgamma)$ by looking at the
development.  The next few lemmas make the correspondence
precise.

\begin{lemma}\lemlab{Z3path}
Let $(G,\bgamma)$ be a $\Z/3\Z$-colored graph, and let $G'$
be a subgraph of $G$.  Then $G'$ has non-trivial $\Z/3\Z$-image
if and only if the lift $\pi^{-1}(G')$ has a path from
some vertex $i_\gamma$ in the fiber over $i\in V(\tilde{G})$
to another vertex $i_{\gamma'}$ in the same fiber.
\end{lemma}
This follows from the fact that $\pi$ is a covering map, but we give a proof for completeness.
\begin{proof}
Assume, w.l.o.g. that $G$ is connected.
Let $C$ be a cycle in $G$.  By \lemref{fundamentalcycles}, the map $\rho$ is defined
completely by its image on the fundamental cycles of a spanning tree $T$ of $G$.  Thus,
by picking a spanning tree for which $C$ is a fundamental cycle (one always exists),
we can recolor $G$ such that all but at most
one of the edges of $C$ has a zero color without changing $\rho$.

With this coloring, it is easy to see that if $\rho(C)=0$ then $\pi^{-1}(C)$ is
three disjoint copies of $C$: each of them contains only vertices $i_{\gamma}$
for a fixed $\gamma\in \{0,1,2\}$.  On the other hand, if $\rho(C)\neq 0$,
let $i\in V(G)$ be a vertex on $C$ and $i_0$ be in the fiber over $i$.  Following the
lift of $C$ in $\tilde{G}$, it will stay on vertices $j_0$ until, when $C$ crosses the
(single) edge with non-zero color, it will leave for a vertex $j_\gamma$,
$\gamma\neq 0$ and then end at $i_{\gamma}$.

Since $C$ was arbitrary, the proof is complete.
\end{proof}

An immediate corollary is
\begin{cor}\corlab{conetwotypes}
Let $(G,\bgamma)$ be a $\Z/3\Z$-colored graph, $G'$ a subgraph of $G$ and $\tilde{G}$ the development.
Then:
\begin{itemize}
\item If $G'$ has trivial image, its lift $\pi^{-1}(G')$ is three disconnected copies of $G'$.
\item If $G'$ has non-trivial image, its lift $\pi^{-1}(G')$ is connected.
\end{itemize}
\end{cor}

\paragraph{Proof of \lemref{orderthree}}
The proof of \lemref{orderthree} is immediate from \lemref{orderthree1}, \lemref{orderthree2}, and
\lemref{orderthree3}, which we prove below.  The proof sketch is:
\begin{itemize}
\item \lemref{orderthree1} says that if the development is a Laman graph, then the
colored graph is a cone-Laman graph.
\item \lemref{orderthree2} says that if the development is not Laman-sparse, then the
colored graph is not cone-Laman sparse.  This is the more difficult implication, since
a violation of Laman-sparsity in $\tilde{G}$ need not coincide with its orbit.
\item \lemref{orderthree3} establishes the correspondence between $(2,3)$-components in the
development that coincide with their orbits and cone-Laman components of the
colored graph.
\end{itemize}

We now state and prove the key lemmas.
\begin{lemma}\lemlab{orderthree1} Let $G$ be a $\Z/3\Z$-colored graph and
$\tilde{G}$ be the development. If $\tilde{G}$ is Laman-sparse then $(G,\bgamma)$
is cone-Laman-sparse.
\end{lemma}
\begin{proof}
We prove the contrapositive.  Assume that $(G,\bgamma)$ is not cone-Laman sparse. There are two ways this can happen, and we check both cases.

\textbf{Case I:} The first case is when $G$ has a subgraph $G'$ with trivial image, $n'$ vertices, and $m'\ge 2n'-2$ edges.  By \corref{conetwotypes},
$\pi^{-1}(G')$ is three copies of $G'$, each of which violates Laman sparsity in $\tilde{G}$.

\textbf{Case II:} Otherwise, $G$ has a subgraph $G'$ with non-trivial image, $n'$ vertices and at least $2n'$ edges.  \corref{conetwotypes}
implies that $\pi^{-1}(G')$ is connected and coincides with its orbit.  Thus, $\pi^{-1}(G')$ has $3n'$ vertices and at least $6n'$ edges, again
violating Laman sparsity in $\tilde{G}$.
\end{proof}

\begin{lemma}\lemlab{orderthree2}
Let $G$ be a $\Z/3\Z$-colored graph and
$\tilde{G}$ be the development. If $\tilde{G}$ is not Laman-sparse then $G$
is not cone-Laman-sparse.
\end{lemma}
\renewcommand{\O}{\mathcal{O}}
\begin{proof}

We prove the contrapositive.  Suppose that $\tilde{G}$ is not Laman-sparse.  Then
it contains a Laman circuit $\tilde{G}'$; let $\mathcal{O}$ be its orbit.
We will show that $\pi(\O)$ violates cone-Laman sparsity in $(G,\bgamma)$.

There are two cases to consider, by \corref{conetwotypes}.

\textbf{Case I:} If $\O$ is three copies of $\tilde{G}'$, then
$\pi(\O)$ is also a copy of $\tilde{G}$, with trivial image.  This violates
cone-Laman sparsity.

\textbf{Case II:} Otherwise, $\O$ is connected, and thus $\alpha_{\gamma}(\tilde{G}')$ and
$\alpha_{\gamma+1}(\tilde{G})$ have non-empty intersection for all $\gamma\in \{0,1,2\}$.
We define $A$ to be $\tilde{G}\cap \alpha_{1}(\tilde{G}')$, and note that all the
pairwise intersections are isomorphic to $A$.  Define $B$ to be
$\tilde{G}'\cap \alpha_1(\tilde{G}')\cap \alpha_2(\tilde{G')}$.  Inclusion-exclusion
shows that
\begin{equation}\label{oedges}
|E(\O)| = 3|E(\tilde{G}')| - 3|E(A)| + |E(B)|
\end{equation}
Here is the key step: since $\tilde{G}'$ is a Laman circuit, $A$ and $B$ are both Laman-sparse.  Thus, the
right-hand-side is minimized when $A$, and $B$, if non-empty, are $(2,3)$-tight: $B$ is a subgraph of $A$, so
adding edges to $A$ or $B$ contributes a negative amount to the r.h.s. of \eqref{oedges}.
In this case, plugging into \eqref{oedges} shows that, if $\O$ has $n'$ vertices, it has exactly $2n'$ edges.

Finally, \corref{conetwotypes} says that $\pi(\O)$ has non-trivial image, and we showed above that
it violates $(2,1)$-sparsity.  This concludes the second case and the proof.
\end{proof}

\begin{lemma}\lemlab{orderthree3} Let $(G,\bgamma)$ be a $\Z/3\Z$-colored graph.
A subgraph $G'$ of $G$ is a cone-Laman rigid component if and only if $\pi^{-1}(G')$ is a symmetric $(2,3)$-component of $\tilde{G}$.
\end{lemma}

\begin{proof}
By \lemref{orderthree1} and \lemref{orderthree2}, a subgraph $G'$ of $G$
is a cone-Laman block if and only if its lift $\pi^{-1}(G')$ is symmetric and a
Laman block in the development $\tilde{G}$.  Maximality of rigid components
now implies the statement.
\end{proof}

With these lemmas, the proof of \lemref{orderthree} is complete.  \hfill $\square$

\end{document}